\documentclass[10pt]{article}
\usepackage{path,graphicx}
\usepackage{a4,amssymb, dsfont}
\usepackage{authblk}
\usepackage{amsmath}
\usepackage{lineno}
\usepackage{enumerate}
\usepackage{cite}
\usepackage[]{algorithm2e}
\usepackage[T1]{fontenc}
\usepackage[utf8]{inputenc}
\usepackage{authblk}
\usepackage[margin=1.1in]{geometry}
\usepackage[justification=centering]{caption}
\newtheorem{definition}{Definition}[section]

\newtheorem{theorem}[definition]{Theorem}
\newtheorem{corollary}[definition]{Corollary}

\newtheorem{claim}[definition]{Claim}

\newtheorem{remark}[definition]{Remark}

\newcommand{\qed}{\hfill $\square$\smallskip}

\title{Fake news and rumors: a trigger for proliferation or fading away}
\author{Ahad N. Zehmakan\thanks{Corresponding author; Email Address: abdolahad.noori@inf.ethz.ch, Postal Address: CAB G 39.3, Institute of Theoretical Computer Science, ETH Z\"urich, Universit\"atstrasse 6, CH-8092 Z\"urich.}}
\affil{Department of Computer Science, ETH Zurich}
\author{Serge Galam}
\affil{CEVIPOF-Centre for Political Research, Sciences Po and CNRS}
\providecommand{\keywords}[1]{\textbf{\textit{Key Words:}} #1}
\date{} 
\begin{document}
\maketitle
\begin{abstract}
The dynamics of fake news and rumor spreading is investigated using a model with three kinds of agents who are respectively the Seeds, the Agnostics and the Others. While Seeds are the ones who start spreading the rumor being adamantly convinced of its truth, Agnostics reject any kind of rumor and do not believe in conspiracy theories. In between, the Others constitute the main part of the community. While Seeds are always Believers and Agnostics are always Indifferents, Others can switch between being Believer and Indifferent depending on who they are discussing with. The underlying driving dynamics is implemented via local updates of randomly formed groups of agents. In each group, an Other turns into a Believer as soon as $m$ or more Believers are present in the group. However, since some Believers may lose interest in the rumor as time passes by, we add a flipping fixed rate $0<d<1$ from Believers into Indifferents. Rigorous analysis of the associated dynamics reveals that switching from $m=1$ to $m\ge2$ triggers a drastic qualitative change in the spreading process. When $m=1$ even a small group of Believers may manage to convince a large part of the community very quickly. In contrast, for $m\ge 2$, even a substantial fraction of Believers does not prevent the rumor dying out after a few update rounds. Our results provide an explanation on why a given rumor spreads within a social group and not in another, and also why some rumors will not spread in neither groups.

\end{abstract}

\section{Introduction}
Last years have witnessed the emergence of the new phenomenon denoted ``fake news'', which has become a worldwide major concern for many actors of political life. In particular, the impact of fake news on twisting democratic voting outcomes has been claimed repeatedly to explain unexpected voting outcomes as Brexit and Trump victories. Fake news have been also identified during the 2017 presidential French campaign. 

Fake news has turned to an important form of social communications, and their spread plays a significant role in a variety of human affairs. It can have a significant impact on people lives, distorting scientific facts and establishment of conspiracy theories.. 

That has triggered the temptation in many countries to curb actual total and anonymous free speech in Internet by setting up new regulations to hinder the political influence of fake news. However, up to date no solid evidence has been found to demonstrate that fake news have indeed reverse a voting outcome. It seems, that fake news are spreading among people having already made their political choice. In any case, to fight against fake news phenomena using the implementation of new judicial regulations it is of importance to understand the mechanisms at work in their spreading. What eventually matters is not the production of fake news but their capacity to spread quickly and massively. It is therefore essential to understand the spreading dynamics.

However, before challenging that goal we claim that the novelty of fake news is not their nature and content but the support of its spreading, namely Internet and social networks. Indeed, fake news are identical to the old category of rumors, which were and still are spread through the process of word-of-mouth. Fake news have created a qualitative change of scale with respect to rumors thank to the new technologies, which have democratized the production and reproduction of information, increasing drastically the rate at which misinformation can spread. Therefore, control and possible handling to manipulate information are now major issues in social organizations including economy, politics, defense, fashion, even personal affairs. The issue of fake news (rumors) spreading has become of a strategic importance at all levels of society. 

Accordingly, we feel that to reach the core of the phenomenon it is more feasible to focus on the phenomenon of rumor spreading carried on by the word-of-mouth. We address the issue by building a mathematical model along an already rich path of contributions from researchers with wide spectrum of backgrounds, like political sciences, statistical physics, computer science, and mathematics. 

%Despite the considerable amount of effort and attention devoted to the study of different rumor spreading models, our understanding is still quite limited. 

In our proposed model, three kinds of agents are considered, the Seeds, the Agnostics and the Others.
%In the present paper we propose a novel rumor spreading model using three kinds of agents are considered, the Seeds, the Agnostics and the Others. 
%like bootstrap percolation~\cite{adler1991bootstrap}, push/pull protocols~\cite{doerr2012rumors}, Galam model~\cite{galam2003modelling}, DK model~\cite{daley1965stochastic} and so on.
The \emph{Seeds} are the initial spreaders of the rumor. They are adamantly convinced the rumor is true and no argument can make them renouncing to it. They trigger the spreading and keep on trying to convince Others to believe in the rumor.
The \textit{Others} constitute the main body of the community. They are either \textit{Believers}, people who are convinced the rumor is true, or \textit{Indifferents}, people who do not believe the rumor is true. Indifferents can turn into Believers if given convincing counter arguments. On the other hand, over time Believers might lose their belief and become Indifferents. 
%At the starting of the rumor, the \textit{Others} are unaware of it. Once confronted to the rumor they may turn to either \textit{Believers}, people who are convinced the rumor is true, or \textit{Indifferents}, people who do not believe the rumor is true. However, \textit{Believers} can turn into \textit{Indifferents} if given convincing counter arguments. Similarly,  \textit{Indifferents} can become  \textit{Believers}. Moreover, over time \textit{Believers} might lose their belief and become \textit{Indifferents}. 
The \emph{Agnostics} are not concerned about the rumor. In addition, they do not believe in any rumors and reject conspiracy theories as a whole. They are always Indifferents to any rumor . 

The dynamics is studied rigorously using local updates in groups formed randomly through informal social gatherings occurring in offices, houses, bars, and restaurants at time break meetings like lunch, happy hours and dinner. Outcomes of those local updates depend on the number $m$ of actual Believers in the group. A Seed contribute to the update as a Believer while an Agnostic acts as an Indifferent. The faith of the spreading is found to depend drastically on the value of $m$ in addition to the actual proportions of Seeds and Agnostics.

The rest of the paper is organized as follows. In Section~\ref{model}, we introduce our model. Section 3 reviews prior models of rumors propagation with which our model shares some attributes. Finally, in Section~\ref{results} we provide our results by the rigorous analysis of the model.

\section{The $\textbf{m}$-rumor spreading model}
\label{model}
In \emph{$m$-rumor spreading model}, we consider a community with $N$ agents, which includes proportions of respectively $s$ of Seeds, $bN$ of Agnostics, and thus $(1-b)N-s$ Others. At each discrete-time round $t\geq 1$, agents gather randomly in fixed rooms of different sizes from 2 to a constant $L\ge 2$, where the numbers of seats in all rooms sum up to $N$, and then the agents update their opinion as follows:
\begin{itemize}
\item An Other who is Indifferent and is in a room with at least $m$ ($m\ge 1$) Believers becomes a Believer 
%\item Each Other in a group with at least $m$ ($m\ge 1$) Believers becomes a Believer.
\item An Other, who is a Believer, becomes Indifferent independently with a probability $0<d<1$.
\item An Agnostic remains always an Indifferent.
\item A Seed remains always a Believer.
\end{itemize}

Figure~\ref{fig1} illustrates one round of the process where a black/white square corresponds to a Seed/Agnostic and a black/white circle corresponds to a Believer/Indifferent. 

\begin{figure}[!ht]
\centering
\includegraphics[scale=0.7]{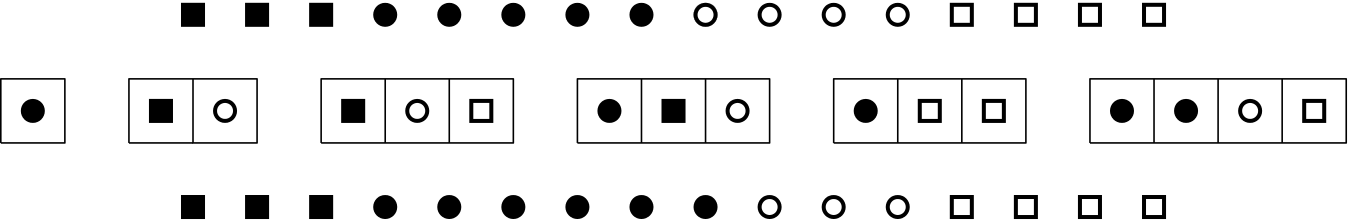}
\caption{One round of the process for $m=2$, $N=16$, $b=0.25$, and $d=0.1$. Two Indifferents turn Believers, but one Believer becomes Indifferent for $d=0.1$.}
\label{fig1}
\end{figure}
Assume that initially $N_0$ individuals are Believers. We are interested in the sequence $N_1, N_2,\cdots$, where the random variable $N_t$ for $t\ge 1$ is the number of Believers after round $t$. 
 
\subsection{Assumptions} \label{assumptions} The value of $m$ is a function of different parameters, like the persuasiveness of the rumor or the tendency of the community in believing that. We assume that $m$ is a fixed positive integer. Recall that there are $s$ Seeds, who initiate the rumor. We assume that $s$ is at least as large as $m$ since $m$ Believers are needed to turn an Indifferent to a Believer. Furthermore, we always suppose that $s$ is a constant while let $N$ tend to infinity. That is, Seeds are a very small group. Moreover, $b$ is assumed to be a small constant, say $0.25$ which implies that $25\%$ of the community are Agnostics. We also assume that $d$ is a fixed small constant, say $0.1$, which means that a Believer loses interest in the rumor after $10$ rounds in expectation and becomes an Indifferent. In the present paper, our main focus is devoted to demonstrate how switching from $m=1$ to $m\ge 2$ triggers a drastic qualitative change in the spreading process, which consequently shed some light on the outcome of some real-world elections. Thus for our purpose, it is realistic to assume that $b\le 0.25$ and $d\le 0.1$.\footnote{We should mention that the values $0.25$ and $0.1$ are chosen to make our calculations more straightforward; otherwise, our results hold also for larger values of $b$ and $d$, say $b\le 0.3$ and $d\le 0.2$. However, this does not apply to any value of $d$ and $b$ as we discuss in Section~\ref{illustration}.} However, it is definitely desirable to analyze the $m$-rumor spreading model for larger values of $d$ and $b$. We provide some intuition and analytic explanations in Section~\ref{illustration}, but the rigorous analysis of the process in this setting is left for future work.

For the sake of simplicity we suppose that all rooms are of size $r$ for some constant $r$. However, our results carry on the general setting with rooms of different sizes by applying basically the same proof ideas. 

Finally, we assume that $r>m$ because otherwise no Indifferent could become a Believer. Note that for an Indifferent to become a Believer, it must share a room with at least $m$ Believers which is not possible for $r\le m$.
\subsection{Markov Chain}
The $m$-rumor spreading model defines a Markov chain on state space $S=\{j:s\leq j\leq (1-b)n\}$, where state $j$ corresponds to having $j$ Believers. (Note that we rule out states $j<s$ and $j>(1-b)n$ since there are $s$ Seeds and $bN$ Agnostics). Furthermore, the transition probability $P_{jj'}$, for $j,j'\in S$, is the probability of having $j'$ Believers in the next round given there are $j$ Believers in the current round. 

Assume that we have $j$ Believers for some $s\le j<(1-b)N$. With some non-zero probability, in the next round at least one Indifferent becomes a Believer and all Believers remain unchanged. Thus, there is a non-zero probability to reach state $(1-b)N$ from state $j$ for any $s\le j\le (1-b)N$. Furthermore, since with some non-zero probability all Believers choose to become Indifferent in the next round, except Seeds\footnote{This is trivial since Seeds remain Believers forever by definition. Thus, from now on we will avoid mentioning it.}, there is a non-zero probability to reach state $s$ from state $j$ for any $s\le j\le (1-b)N$. Therefore, the process eventually reaches state $s$, where only Seeds believe in the rumor, or it reaches some state $j\ge N/2$, where at least half of the community believe in the rumor. We say the rumor \emph{dies out} in the first case and we say it \emph{takes over} in the second one.

\subsection{Our contribution}
Having the $m$-rumor spreading model in hand, two natural questions arise: what are the conditions for which a rumor takes over or dies out? and how fast does this happen? The main goal of the present paper is to address these two basic questions.

First we consider the $m$-rumor spreading process for $m=1$ and prove that if $N_0\ge s$, then the rumor takes over in $\mathcal{O}(\log N)$ rounds asymptotically almost surely (a.a.s.)\footnote{We say an event happens asymptotically almost surely whenever it happens with probability $1-o(1)$ while we let $N$ tend to infinity.}. This implies that if even initially only Seeds, which are a group of constant size, believe in the rumor, for $m=1$ the rumor takes over in logarithmically many rounds. However, by switching from $m=1$ to $m\ge 2$ a very different picture emerges. We prove that for $m\ge 2$ the rumor dies out in $\mathcal{O}(\log N)$ rounds a.a.s. if $N_0\le \alpha N$ for some sufficiently small constant $\alpha>0$. This means if even initially a constant fraction of the community are Believers, for $m\ge 2$ the rumor dies out. Therefore, a drastic change occurs in the behavior of the process by switching from $m=1$ to $m\ge 2$, i.e., the process exhibits a threshold behavior. 
%Furthermore, we show that the upper bound of $\mathcal{O}(\log n)$ is asymptotically tight. 

Let us illustrate this threshold behavior by a numerical example. Consider a community of size $N=10^6$ and assume that $b=0.25$, $d=0.1$, $r=3$, and $s=100$. Let $f(n)$ denote the expected number of Believers in the next round by assuming that there are $n$ Believers in the current round (we will provide an exact formula for $f(n)$ in Section~\ref{results}, Equation~(\ref{eq 1})). One can observe in Figure~\ref{fig2} (left) that $f(n)>n$ for $s=100\le n\le 500,000=N/2$. Therefore, by starting from 100 Believers, in each round the number of Believers increases in expectation until there are $N/2$ Believers, i.e., the rumor takes over. In Figure~\ref{fig2} (right), we present the expected percentage of Believers in the $t$-th round for $0\le t\le 20$ by assuming that $N_0=100$. It demonstrates that if initially only $0.01\%$ of the community are Believers, the expected percentage of Believers after 15 rounds is almost equal to $67.5\%$, which implies that we expect the rumor to take over in less than 15 rounds. Note that the Believers can constitute at most $75\%$ of the community, but in that case $10\%$ of them, which is $7.5\%$ of the whole community, would turn into Indifferent in expectation since $d=0.1$. This may explain where the value $67.5\%$ comes from.
\begin{figure}[!ht]
\centering
\includegraphics[scale=0.5]{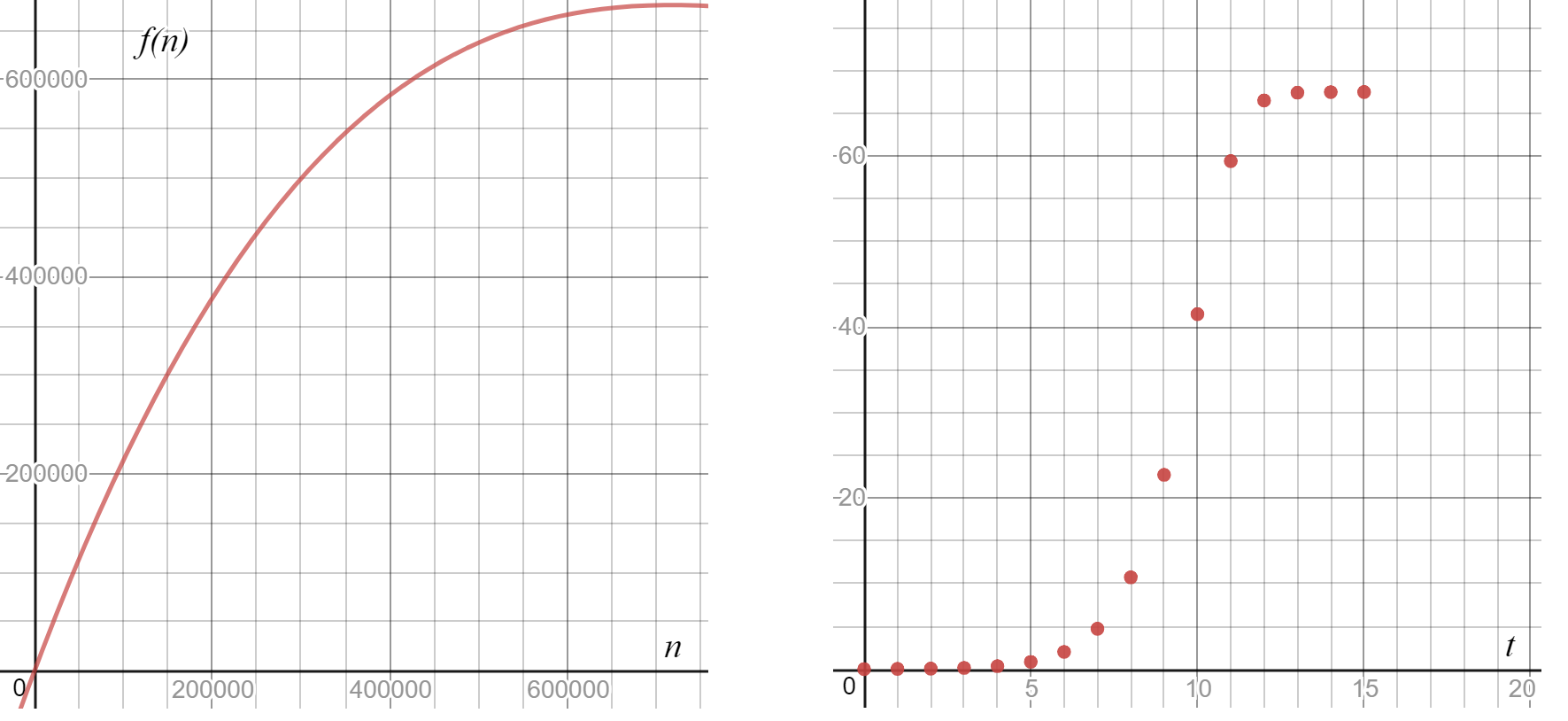}
\caption{(left) Variation of $f(n)$ (right) the expected percentage of Believers in round $t$ given $N_0=100$.}
\label{fig2}
\end{figure}

Now, we consider the same setting but for $m=2$. Figure~\ref{fig3} (left) illustrates that $f(n)<n$ for $s=100\le n\le 100,000=N/10$. We also provide the variation of function $f(n)-n$ in Figure~\ref{fig3} (middle), where it is easier to spot that $f(n)<n$. Thus, if initially even $10\%$ of the community believe in the rumor, we expect the number of Believers decreases in each round until only the $100$ Seeds believe in the rumor, i.e., it dies out. In Figure~\ref{fig3} (right), the expected percentage of Believers in the $t$-th round for $0\le t\le 69$ are drawn by assuming that $N_0=100,000$. This explains that if initially even $10\%$ of the community are Believers, the rumor still dies out in 69 rounds in expectation. 
\begin{figure}[!ht]
\centering
\includegraphics[scale=0.4]{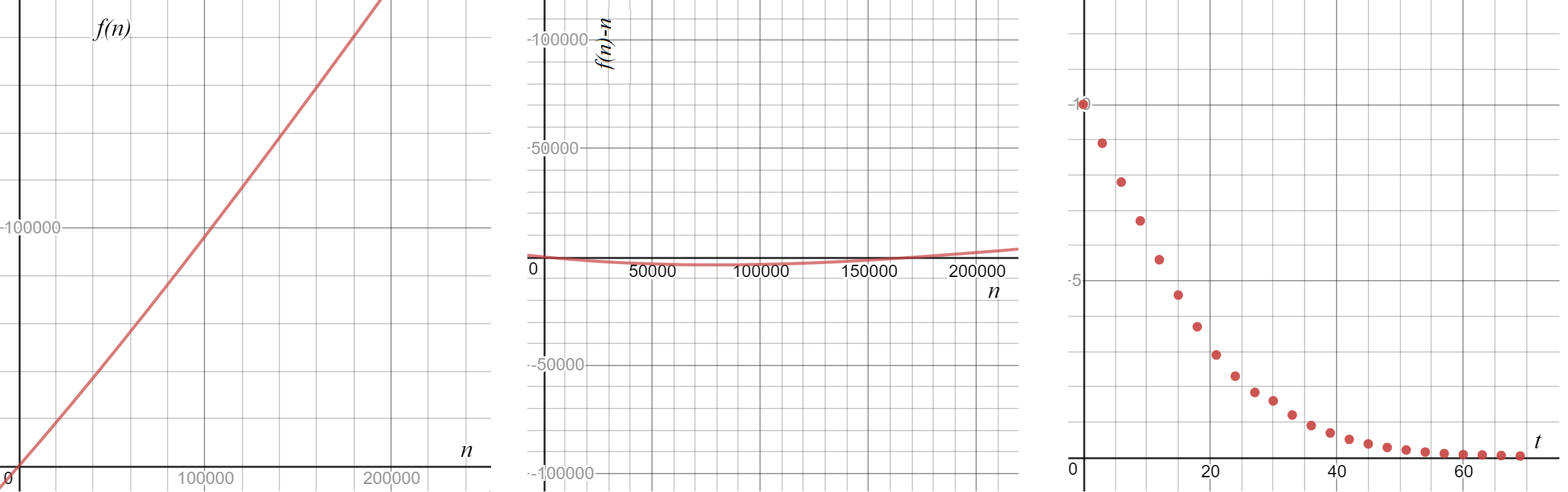}
\caption{(left) Variation of $f(n)$ (middle) variation of $f(n)-n$ (right) the expected percentage of Believers in round $t$ given $N_0=100,000$.}
\label{fig3}
\end{figure}
%\begin{linenomath*}
%\begin{align*}
%& 100\rightarrow 250\rightarrow 610\rightarrow 1473\rightarrow 3539\rightarrow 8469\rightarrow 20,139\rightarrow47,236\rightarrow107,346\rightarrow\\
%& 227,189\rightarrow415,049\rightarrow593,896\rightarrow664,876\rightarrow673,962\rightarrow 674,531\rightarrow 674,563\\
%& \rightarrow 674,564^\curvearrowleft
%\end{align*}
%\end{linenomath*}

%\begin{linenomath*}
%\begin{align*}
%& 100,000\rightarrow 96,510\rightarrow 92,956\rightarrow 89,347\rightarrow 85,696\rightarrow 82,015\rightarrow 78,316\rightarrow\cdots\\
%& 23,174\rightarrow 21,257 \rightarrow 19,470\rightarrow 17,809\rightarrow 16,270\rightarrow 14,847\rightarrow 13,534\rightarrow\\
%& 12,325\rightarrow 11,214\rightarrow 10,195\rightarrow 9262\rightarrow 8409\rightarrow\cdots\rightarrow 100^\curvearrowleft
%\end{align*}
%\end{linenomath*}

Let us briefly discuss the intuition behind such a change in the behavior of $m$-rumor spreading process. Assume that initially Believers constitute a small fraction of the community. In that case, it is a very unlikely that a room includes two or more Believers. Therefore, if $m=2$, then almost no ``new'' Believer is generated, i.e., no Indifferent become a Believer. Furthermore, each Believer becomes Indifferent with probability $d$. Thus, we expect the number of Believers to decrease by a constant factor. On the other hand, if $m=1$, then all Believers who are in different rooms can single-handedly turn all the agents in their room into Believers. Thus, the number of Believers increases by an $r$ factor. (Actually, this is a bit smaller since we did not take Agnostics into account.) Note that in this case also a $d$ fraction of the Believers become Indifferent in expectation, but this is negligible since we assume that $d$ is in order of $0.1$ but $r\ge 2$. Therefore, the number of Believers increases by a constant factor, in expectation.

\subsection{Proof techniques.} 
We analyze the model rigorously applying standard tools and techniques from probability theory. We compute the formula for the expected number of Believers as a function of the model parameters and the number of Believers in the previous round. This formula turns out to be quite involved. Thus, instead we work with suitable upper and lower bounds, which are easier to handle. Building on these bounds, we show that if $m=1$, in expectation the number of Believers increases by a constant factor in each round until the rumor takes over. To turn such an expectation based argument into an a.a.s. statement, we exploit classical concentration inequalities, like the Chernoff bound and Azuma's inequality. Since the number of Believers increases by a constant factor in each round, the rumor takes over in $\mathcal{O}(\log N)$ rounds. We prove that this bounds is asymptotically tight; that is, for some cases a.a.s. the rumor needs $\Omega(\log N)$ rounds to take over. A similar argument holds for the case of $m\ge 2$, where in each round the number of Believers decreases by a constant factor until the rumor dies out.

\subsection{Real cases ground} \label{real cases} Within our model, for a given rumor, its success or failure in spreading depends merely on the value of $m$. The value $m=1$ drives viral the rumor while it fades quickly for $m\ge 2$. The actual value of $m$ is expected to be a function of both the rumor content and the social characteristics of the community in which it is launched. Therefore the same rumor may spread in a community and fade away in another. And of two different rumors within a community, one may spread and the other vanishes.

In \cite{galam2003modelling} those differences were given an explanation in terms of the existence of different prejudices and cognitive biases, which were activated at ties in even discussing groups with otherwise a local majority rule update. Here we provide another explanation for this phenomena, based on the value of $m$ in the discussing group, which does not obey a local majority rule for shifting opinions. It is worth to mention that our current approach could in principle be extended to account for a mixture of arguments with different convincing power leading to a combination of groups with $m=1$ and $m\ge 2$. The $m=1$ case embodies the promptness of an agent to believe in a given rumor or to the strength of the argument. The same agent may requires $m\ge 2$ to adopt another rumor as true either due to weaker arguments or the content of the rumor.

The fact that the actual value of $m$ can vary from one social group to another with respect to the same rumor can be illustrated using some recent real cases of differentiated rumor spreading. One significant case relates to the terrorist attack in Strasbourg, France, on the eve of December 11, 2018. The terrorist, a follower of Islamic State, claimed five person lives and wounded eleven persons, before being shot by the police a few days latter. With the shooting occurring amid the uprising of the Yellow Jackets (Gilets Jaunes) movement, a conspiracy theory emerged at once and spread quite quickly among a specific community of French people, namely the ones identifying with the Yellow Jacket movement, claiming the attack was set up by the French government to distract support to the on going social movement. At the same time the rumor did not take over among the rest of the French population as shown by the following poll figures. The statement? Evidences are not clear about who committed the attack or the attack was a set up by French government? was found to be agreed on by 42\% of yellow jacket activists and 23\% of yellow jacket supporters against only 11\% within non supporters of yellow jackets. 
%(https://en.wikipedia.org/wiki/2018_Strasbourg_attack)

A similar feature is exhibited with an American case about supposed Russian interference in last American Presidential election. More specifically, attributing the hacking of Democratic e-mails to Russia was supported by 87\% of Clinton voters against only 20\% of Trump voters. 
%(https://today.yougov.com/topics/politics/articles-reports/2016/12/27/belief-conspiracies-largely-depends-political-iden2016)

\section{Prior Work}
\label{prior work}
As mentioned, the question of how a rumor spreads in a community has been studied extensively in different contexts and numerous models have been introduced and investigated, both theoretically and experimentally, see e.g.~\cite{nowak2019homogeneous,pires2017dynamics,wang2017rumor,merlone2014reaching,borge2013emergence,rabajante2011mathematical}. We shortly discuss some results regarding Galam model, DK model, and bootstrap percolation model, which are arguably the closest ones to ours. 

\subsection{Bootstrap percolation model} 
Consider the relationship network between the people in a community and assume that initially each agent is a Believer or an Indifferent toward a rumor. In $m$-bootstrap percolation, an Indifferent becomes a Believer as soon as $m$ of its connections in the network are Believers and remains a Believer forever. The main question in this context is to determine the minimum number of agents who must be Believers initially to guarantee that the rumor takes over the whole network eventually (that is, all agents become Believers). This has been studied on different network structures, like hypercube~\cite{balogh2010bootstrap}, lattice~\cite{jeger2019dynamic}, random graphs~\cite{gartner2018majority,n2018opinion}, and many others. In contrast to the bootstrap percolation model, which assumes that the agents are sitting on the nodes of an underlying network, in our model agents gather in groups of different sizes at random. In this sense, our model is more similar to Galam model, introduced in the next section.

\subsection{Galam model}
The Galam model was introduced by Galam~\cite{galam2002minority} to describe the spreading of a minority opinion resulting in democratic rejection of social reforms initially favored by a majority. In this model, initially each agent is positive or negative regarding a reform proposal. Then, in each discrete-time round they are distributed randomly in groups of different sizes and everyone adopts the most frequent opinion in its group; in case of a tie, negative opinion is chosen. Building on this model, Galam~\cite{galam2002minority} provided some illustrations regarding the output of some real-world elections. Later, several extensions of the model were proposed and studied, for example by adding contrarian effect~\cite{galam2004contrarian}, introducing random tie-breaking rule~\cite{galam2005heterogeneous,gartner2017phase}, considering three competing opinions~\cite{gekle2005opinion}, and defining the level of activeness~\cite{qian2015activeness}. 

\subsection{DK model}
Another well-established rumor spreading model is DK model, introduced by Daley and Kendall~\cite{daley1965stochastic}. There are essentially two differences between DK model and our model. Firstly in DK model, it is assumed that the gatherings are always of size two. Furthermore, a Believer stops spreading the rumor as soon as it encounters another Believer, in DK model. The illustration for such an updating rule is that the agent decides that it is no longer ``news'', and thus stops spreading it. In that sense, our model is perhaps more similar to the variant introduced by Zhao et al.~\cite{zhao2011rumor}, where a Believer becomes Indifferent independently with some probability $d>0$. 
\section{Rigorous demonstration of the $\textbf{m}$-rumor spreading}
\label{results} 
In this section we rigorously analyze the behavior of the $m$-rumor spreading model. The goal is to prove Theorem~\ref{theorem 1} which is the main contribution of the present paper.

\begin{theorem}
\label{theorem 1}
In the $m$-rumor spreading model
\begin{enumerate}[(i)]
\item  for $m=1$ if $N_0\ge s$, the rumor takes over
\item for $m=2$ if $N_0\le \alpha N$, the rumor dies out
\end{enumerate}
asymptotically almost surely in $\mathcal{O}(\log N)$ rounds, where $\alpha>0$ is a sufficiently small constant. 
\end{theorem}

Theorem~\ref{theorem 1} states that by switching from $m=1$ to $m\ge 2$, the process behaves substantially differently. For $m=1$, the rumor takes over if even Others are all Indifferent initially. On the other hand, for $m=2$ the rumor dies out even though a constant fraction of the community believe in the rumor already. They both occur in $\mathcal{O}(\log N)$ rounds, which is shown to be asymptotically tight.

We first set up some basic definitions and state some standard concentration inequalities in Section~\ref{preliminaries}. Then, we provide the proof of Theorem~\ref{theorem 1} part (i) and part (ii) respectively in Sections~\ref{m=1} and~\ref{m=2}.

\subsection{Preliminaries}
\label{preliminaries}
Recall that random variable $N_t$ for $t\ge 0$ denotes the number of Believers at the end of round $t$. The conditional random variable $N_t|N_{t-1}=n$, called ``$N_t$ given $N_{t-1}=n$'', has probability mass function $Pr[N_t^n=n']=Pr[N_t=n'|N_{t-1}=n]$, where we shortly write $N_t^n$ for $N_t|N_{t-1}=n$. Furthermore, $\mathbb{E}[N_t^n]$ denotes the expected value of random variable $N_t^n$.

For $m=1$, we show that $\mathbb{E}[N_t^n]>\beta n$, for some constant $\beta>1$, if $s\le n\le N/2$. Therefore, one expects the rumor to take over in logarithmically many rounds. Furthermore for $m=2$, $\mathbb{E}[N_t^n]<\beta' n$, for some constant $\beta'<1$, if $s< n\le \alpha N$. Thus, by starting from $\alpha N$ Believers we expect the number of Believers decreases by a constant factor in each round, i.e., the rumor dies out in a logarithmic number of rounds. This should intuitively explain why Theorem~\ref{theorem 1} holds. However, to turn this expectation based argument into a formal proof, we need to show that random variable $N_t^n$ is sharply concentrated around its expectation and apply some careful calculations. In the rest of this section, we first compute the value of $\mathbb{E}[N_t^n]$ and then provide some basic concentration inequalities, which are later used to prove that $N_t^n$ is concentrated around $\mathbb{E}[N_t^n]$.

To calculate $\mathbb{E}[N_t^n]$, let us first determine the probability that an Indifferent agent $I$ becomes a Believer~\footnote{Whenever we talk about an Indifferent becoming a Believer, we clearly mean an Other who is an Indifferent; otherwise, an Agnostic never becomes a Believer by definition.} in the $t$-th round given $N_{t-1}=n$. We claim that this probability is equal to
\begin{equation}
\label{eq 8}
N\sum_{j=m}^{r-1}{n\choose j}{N-n-1\choose r-j-1}(r-1)!(N-r)!/N! .
\end{equation}
There are $N!$ possibilities to assign $N$ agents to $N$ seats. Let us call each of these assignments a \emph{permutation}. In Equation~(\ref{eq 8}),
the value $N!$ in the denominator stands for the number of all possible permutations and the numerator is equal to the number of permutations for which agent $I$ becomes a Believer. There are $N$ possibilities for fixing $I$'s seat, and the sum corresponds to the number of possibilities to assign the remaining $N-1$ agents such that there are at least $m$ Believers in $I$'s room. For that, one can choose exactly $j$ Believers, for some $m\le j\le r-1$, among $n$ Believers to share the room with agent $I$, and the remaining seats in the room will be taken by $r-j-1$ agents among the $N-n-1$ Indifferents. After the selection of the $r-1$ agents who are in the same room with $I$, there are $(r-1)!$ possibilities of assigning them to the seats and also $(N-r)!$ possibilities to assign the other agents to the $(N-r)$ remaining seats. 

Each Believer remains a Believer independently with probability $1-d$ if it is an Other and with probability 1 if it is one of the $s$ Seeds. Furthermore, each of the $(1-b)N-n$ Indifferents becomes a Believer with the probability given in Equation~(\ref{eq 8}) (while each of the $bN$ Agnostics remains an Indifferent). Therefore, we have
\begin{equation}
\label{eq 1}
\mathbb{E}[N_t^n]=s+(n-s)(1-d)+((1-b)N-n)\frac{(r-1)!\sum_{j=m}^{r-1}{n\choose j}{N-n-1\choose r-j-1}}{\prod_{j=1}^{r-1}(N-j)}.
\end{equation}

The above formula is a bit involved; thus, we sometimes will utilize relatively simpler formulas, which are easier to handle, as lower and upper bounds in our proofs.

Let us recall three basic concentration inequalities which we will apply repeatedly later.

\begin{theorem} (Markov's inequality~\cite{dubhashi2009concentration})
\label{Markov}
Let $X$ be a non-negative random variable and $a>0$, then
\[
Pr[X\ge a]\le \frac{\mathbb{E}[X]}{a}.
\]
\end{theorem}
\begin{theorem}(Chernoff bound~\cite{dubhashi2009concentration})
\label{Chernoff}
Suppose $x_1,\cdots,x_k$ are independent Bernoulli random variables taking values in $\{0,1\}$ and let $X$ denote their sum, then for $0\leq \epsilon\leq 1$
\[
Pr[(1+\epsilon)\mathbb{E}[X]\leq X]\leq \exp(-\frac{\epsilon^2\mathbb{E}[X]}{3})
\]
and
\[
Pr[X\leq (1-\epsilon)\mathbb{E}[X]]\leq \exp(-\frac{\epsilon^2\mathbb{E}[X]}{2}).
\]
\end{theorem}
\begin{theorem}(Azuma's inequality~\cite{dubhashi2009concentration})
\label{Azuma}
Let $x_1,\cdots, x_k$ be an arbitrary set of random variables and let $f$ be a function satisfying the property that for each $i\in[k]:=\{1,\cdots,k\}$, there is a non-negative $c_i$ such that
\[
|\mathbb{E}[f|x_i,\cdots, x_1]-\mathbb{E}[f|x_{i-1},\cdots, x_1]|\le c_i.
\]
Then, for constant $\epsilon>0$
\[
Pr[(1+\epsilon)\mathbb{E}[f]< f]\le \exp(-\frac{\epsilon^2\mathbb{E}[f]^2}{2c})=\exp(-\Theta(\frac{\mathbb{E}[f]^2}{c}))
\]
and 
\[
Pr[f<(1-\epsilon)\mathbb{E}[f]]\le \exp(-\frac{\epsilon^2\mathbb{E}[f]^2}{2c})=\exp(-\Theta(\frac{\mathbb{E}[f]^2}{c}))
\]
where $c:=\sum_{i=1}^{k}c_i^2$.
\end{theorem}
Note that although Markov's inequality does not need any sort of independence, the Chernoff bound requires $X$ to be the sum of independent random variables. Azuma'a inequality also requires some level of independence in the sense that to get a reasonable tail bound, the value of $c$ has to be small. 

We define the random variable $N_t^n$ as a function of $N$ Bernoulli random variables and apply Azuma's inequality to attain Corollary~\ref{corollary 1}.

\begin{corollary}
\label{corollary 1}
In the $m$-rumor spreading model, for any $\epsilon>0$
\[
Pr[(1+\epsilon)\mathbb{E}[N_{t}^n]<N_{t}^n]\leq \exp(-\Theta(\frac{\mathbb{E}[N_t^n]^2}{N}))
\]
and 
\[
Pr[N_{t}^n< (1-\epsilon)\mathbb{E}[N_{t}^n]]\le \exp(-\Theta(\frac{\mathbb{E}[N_t^n]^2}{N})).
\]
\end{corollary}
\begin{proof}
We prove the first inequality, and the proof of the second one is analogous. Consider an arbitrary labeling from $1$ to $N$ on the seats. We define Bernoulli random variable $x_i^n(t)$ for $1\le i\le N$ to be 1 if and only if the agent assigned to the $i$-th seat is a Believer given $N_{t-1}=n$. Then, for any $1\le i\leq N$
\[
|\mathbb{E}[N_t^n|x_i^n(t),\cdots, x_1^n(t)]-\mathbb{E}[N_t^n|x_{i-1}^n(t),\cdots, x_1^n(t)]|\le L'
\]
where $L'$ is a constant as a function of $L$. Now by applying Theorem~\ref{Azuma}, we have
\[
Pr[(1+\epsilon)\mathbb{E}[N_{t}^n]<N_{t}^n]\leq \exp(-\Theta(\frac{\mathbb{E}[N_t^n]^2}{\sum_{i=1}^{N}L'^2}))=\exp(-\Theta(\frac{\mathbb{E}[N_t^n]^2}{N})).
\]\qed
\end{proof}

\subsection{Proof of Part (i): $\mathbf{m=1}$}
\label{m=1}
In this section, we prove that for $m=1$ if $N_0\ge s$, then the rumor takes over a.a.s. in $\mathcal{O}(\log N)$ rounds. We divide our analysis into three phases. We show that by starting with $s$ Believers, the process reaches a state with at least $\Omega(\log\log\log N)$ Believers (phase 1), then it reaches a state with at least $N^{2/3}$ Believers (phase 2), and finally at least half of the community will become Believers (phase 3). Furthermore, each of these three phases takes $\mathcal{O}(\log N)$ rounds. As discussed, the main idea of the proof is to show that the number of Believers increases by a constant factor in expectation in each round and then apply the fact that $N_t^n$ is sharply concentrated around its expectation. However, there are two subtle issues which make us to split our analysis into three phases. Firstly, the exact formula for $\mathbb{E}[N_t^n]$ is quite involved, see Equation~(\ref{eq 1}). Therefore, we sometimes have to work with relatively simpler lower/upper bounds, which are technically easier to handle, but hold only for a particular range of $n$. Secondly, the error probability provided by Corollary~\ref{corollary 1} is of form $\exp(-\Theta(\frac{\mathbb{E}[N_t^n]^2}{N}))$, which is equal to a constant for $\mathbb{E}[N_t^n]=\mathcal{O}(\sqrt{N})$. Therefore, for some settings we need to bound the error probability by applying other techniques. Furthermore, we should mention that the values $\log\log\log N$ and $N^{2/3}$ are selected in a way to make the calculations straightforward, otherwise the proof works for some other values as well.

\paragraph{Phase 1.} We prove that by starting from $N_0\ge s$, there are at least $\log_c\log\log N$ Believers after $\mathcal{O}(\log N)$ rounds a.a.s., where $c>1$ is a constant to be determined later.\footnote{We assume that all logarithms are to base $e$, otherwise we point out explicitly.}

Let us first provide Claim~\ref{claim 1}, whose proof is given below.

\begin{claim}
\label{claim 1}
Assume that $T=\log_r\log_c\log\log N$ and $1\le n_0\le \log_c\log\log N$, then for $t'\ge 0$,
\begin{equation*}
Pr[N_{t'+T}\ge r^{T}n_0|N_{t'}=n_0]\geq \frac{1}{\sqrt{\log N}}.
\end{equation*}
\end{claim}
Claim~\ref{claim 1} implies that from a state with $n_0\ge 1$ Believers, after $T$ rounds the process reaches a state with at least $r^{T}n_0\ge r^{T}=\log_c\log\log N$ Believers with probability at least $1/\sqrt{\log N}$. The probability that the process runs for $T\log^{2/3} N$ rounds and never reaches $\log_c\log\log N$ Believers or more is upper-bounded by
\[
(1-\frac{1}{\sqrt{\log N}})^{\log^{2/3} N}\le \exp(-\log^{1/6}N)=o(1)
\]

where we used $1-x\le \exp(-x)$. Thus, for some $t_1\le T\log^{2/3}N\le \log N$ the process reaches at least $\log_c\log\log N$ Believers with probability $1-o(1)$.

\paragraph{Proof of Claim~\ref{claim 1}.} 
Let us first compute $Pr[N_t^n=rn]$, which is the probability that there are $rn$ Believers in the $t$-th round given $N_{t-1}=n$, for some $n\le b N$. We define $\mathcal{E}_1$ to be the event that in round $t$ each Believer is assigned to a room with $r-1$ Others who are all Indifferents given $N_{t-1}=n$. Furthermore, $\mathcal{E}_2$ is the event that no Believer becomes Indifferent in round $t$ given $N_{t-1}=n$. The number of Believers will increase to $rn$ in round $t$ if and only if both events $\mathcal{E}_1$ and $\mathcal{E}_2$ occur. That is,
\begin{equation}
\label{eq1}
Pr[N_t^n=rn]=Pr[\mathcal{E}_1]\cdot Pr[\mathcal{E}_2]\ge Pr[\mathcal{E}_1]\cdot (1-d)^n
\end{equation}
where we used that events $\mathcal{E}_1$ and $\mathcal{E}_2$ are independent and $Pr[\mathcal{E}_2]= (1-d)^{n-s}\ge (1-d)^n$. We claim that
\[
Pr[\mathcal{E}_1]={N/r\choose n}n!\ r^n{(1-b)N-n \choose (r-1)n}((r-1)n)!(N-rn)!/N!
\]
where the value $N!$ in the denominator stands for the number of all possible permutations and the numerator is equal to the number of permutations for which the event $\mathcal{E}_1$ occurs. The value ${N/r\choose n}$ is the number of possibilities of selecting the rooms which contain the Believers. Moreover, there are $n!\ r^n$ ways of placing $n$ Believers in these $n$ rooms such that each of the rooms includes one Believer. We need to choose $(r-1)n$ Indifferents among all $(1-b)N-n$ Others who are Indifferents to fill in the remaining $(r-1)n$ seats in these rooms. There are $((r-1)n)!$ possible assignments of the $(r-1)n$ chosen agents into the $(r-1)n$ seats. Finally, there are $(N-rn)!$ possibilities to place the remaining $N-rn$ agents.

Now, we approximate this probability. Since ${n \choose k}\ge (\frac{n}{k})^k$ by Stirling's approximation~\cite{marsaglia1990new} and $k!\ge (k/e)^k$, we get
\[
Pr[\mathcal{E}_1] \ge \frac{(\frac{N}{rn})^n\ (\frac{n}{e})^n\ r^n\ (\frac{(1-b)N-n}{(r-1)n)})^{(r-1)n}\ (\frac{(r-1)n}{e})^{(r-1)n}}{\prod_{j=0}^{rn-1}(N-j)}.
\]
We have $(1-b)N-n\geq (1-2b)N$ for $n\le b N$. Thus, by applying $\prod_{j=0}^{rn-1}(N-j)\le N^{rn}$ and $b\le 0.25$, we get
\[
Pr[\mathcal{E}_1] \ge \frac{(1-2b)^{(r-1)n}}{\exp(rn)}\ge \exp(-2rn).
\]
Combining this inequality and Equation~(\ref{eq1}) yields
\begin{equation}
\label{eq2}
Pr[N_t^n=rn]\ge \exp(-2rn)\cdot (1-d)^n\ge c^{-n}
\end{equation}
for some constant $c>1$ (which is the constant that we promised to determine later).

The probability $Pr[N_{t'+T}\ge r^{T}n_0|N_{t'}=n_0]$ is equal to
\[
\prod_{t=1}^{T}Pr[N_{t'+t}=r^tn_0|N_{t'+t-1}=r^{t-1}n_0]\ge \prod_{t=1}^{T}c^{-r^{t-1}n_0}\ge c^{-r^{T}n_0T}
\]
where we applied Equation~(\ref{eq2}). (To apply Equation~(\ref{eq2}) we need the inequality $r^{t-1}n_0\le b N$ to hold, which is the case for $t\le T= \log_r\log_c\log\log N$ and $n_0\le \log_c\log\log N$.) Now, by plugging in $n_0\leq \log_c\log\log N$ and $T=\log_r\log_c\log\log N$ implies that
\[
Pr[N_{t'+T}= r^{T}n_0|N_{t'}=n_0]\ge c^{-(\log_c\log\log N)^3}\ge \frac{1}{\sqrt{\log N}}.
\]
This finishes the proof of Claim~\ref{claim 1}.\qed
\paragraph{Phase 2.} So far we proved that a.a.s. after $t_1$ rounds, for some $t_1\leq \log N$, there are at least $n_1=\log_c\log\log N$ Believers. Now, we show that from a state with at least $n_1$ Believers, a.a.s. the process reaches a state with at least $N^{2/3}$ Believers after $\mathcal{O}(\log N)$ rounds. In this phase we always assume that $n_1\le n \le N^{\frac{2}{3}}$.

Let us define $\mathcal{E}_3$ to be the event that at least $\frac{3}{20}n$ Indifferents become Believers in round $t$ given $N_{t-1}=n$. Furthermore, $\mathcal{E}_4$ is the event that at most $\frac{5}{40}n$ Believers become Indifferents in round $t$ given $N_{t-1}=n$. We first provide lower bounds on the probabilities $Pr[\mathcal{E}_3]$ and $Pr[\mathcal{E}_4]$, respectively in Claim~\ref{claim 2} and Claim~\ref{claim 3}.

\begin{claim}
\label{claim 2} $Pr[\mathcal{E}_3]\ge 1-\exp(-\Theta(N))$.
\end{claim}
\begin{proof}
Assume that there are $n$ Believers. Let $R_1$ be the set of permutations where at least $\frac{1}{20}n$ Believers share their room with at least another Believer. Furthermore, define $R_2$ to be the set of permutations where there are at least $\frac{8}{10}n$ Believers who share their room with $(r-1)$ Agnostics. We claim that
\begin{equation}
\label{eq 21}
Pr[\mathcal{E}_3]\ge 1-\frac{|R_1|+|R_2|}{N!}.
\end{equation}
This is true because if a permutation is not in $R_1\cup R_2$, then there are at least $n-(\frac{1}{20}n+\frac{8}{10}n)=\frac{3}{20}n$ Believers who share their room with no Believer but at least one Other who is Indifferent. Therefore, each of these $\frac{3}{20}n$ Believers single-handedly turns an Indifferent into a Believer. That is, at least $\frac{3}{20}n$ Indifferents become Believers.

There are ${n\choose n/20}$ possibilities to choose $\frac{1}{20}n$ Believers. We have ${N/r\choose n/40}$ possibilities for the selection of $\frac{1}{40}n$ rooms. Furthermore, $(\frac{rn}{40})^{\frac{n}{20}}$ is an upper bound on the number of placements of these $\frac{1}{20}n$ Believers in the $\frac{rn}{40}$ seats in the selected rooms. Finally, there are $(N-\frac{n}{20})!$ possibilities to place the remaining agents. Clearly, in this way we count each permutation which is in $R_1$ at least once. Therefore, we have
\[
\frac{|R_1|}{N!}\le \frac{{n\choose n/20}{N/r \choose n/40}(\frac{rn}{40})^{\frac{n}{20}}(N-\frac{n}{20})!}{N!}.
\]
By applying ${n \choose k}\le (\frac{ne}{k})^k$, Stirling's approximation~\cite{marsaglia1990new}, we get
\begin{align*}
& \frac{|R_1|}{N!}\le\frac{(20e)^{\frac{n}{20}}(\frac{40eN}{rn})^{\frac{n}{40}}(\frac{rn}{40})^{\frac{n}{20}}}{\prod_{j=0}^{\frac{n}{20}-1}(N-j)}\le \frac{(20e)^{\frac{n}{20}}(\frac{40eN}{rn})^{\frac{n}{40}}(\frac{rn}{40})^{\frac{n}{20}}}{(\frac{19}{20}N)^{\frac{n}{20}}}.
\end{align*}
In the last step we used the fact that $N-j\ge \frac{19}{20}N$ for $0\le j\le \frac{n}{20}-1$. By some simplifications we get
\begin{equation}
\label{eq 22}
\frac{|R_1|}{N!}\le (\frac{(20e)^2\times e\times r\times 20^2}{40\times 19^2})^{\frac{n}{40}}(\frac{n}{N})^{\frac{n}{40}}\le \exp(-\Theta(n)) 
\end{equation}
where we used that $\frac{n}{N}=o(1)$ for $n\le N^{\frac{2}{3}}$.

To bound $|R_2|$, we select $\frac{8n}{10}$ Believers, $\frac{8(r-1)n}{10}$ Agnostics, and $\frac{8n}{10}$ rooms. This can be done in ${n\choose 8n/10}{bN\choose 8(r-1)n/10}{N/r\choose 8n/10}$ different ways. There are $(\frac{8n}{10})!\ r^{\frac{8n}{10}}$ possibilities to place the Believers in these rooms such that each of the rooms has exactly one Believer. Finally, there are $(\frac{8(r-1)n}{10})!$ possibilities to place the Agnostics in the remaining $\frac{8(r-1)n}{10}$ seats in these rooms and $(N-\frac{8rn}{10})!$ possibilities to place the remaining agents. In this way, we count a permutation in $R_2$ at least once. Therefore, we have
\[
\frac{|R_2|}{N!}\le \frac{{n\choose 8n/10}{bN\choose 8(r-1)n/10}{N/r\choose 8n/10}(\frac{8n}{10})!r^{\frac{8n}{10}}(\frac{8(r-1)n}{10})!(N-\frac{8rn}{10})!}{N!}.
\]
Note that $N-j\ge \frac{19}{20}N$ for $0\le j\le \frac{8rn}{10}-1$ and $n\le N^{\frac{2}{3}}$. This implies that
\[
\frac{(N-\frac{8rn}{10})!}{N!}=\frac{1}{\prod_{j=0}^{\frac{8rn}{10}-1}N-j}\le \frac{1}{(\frac{19}{20}N)^{\frac{8rn}{10}}}.
\]
By applying this inequality and ${n \choose k}\le (\frac{ne}{k})^k$ and $n!\le e\sqrt{n}(\frac{n}{e})^{n}$, Stirling's approximation~\cite{marsaglia1990new}, we get
\begin{align*}
& \frac{|R_2|}{N!}\le \frac{(\frac{10e}{8})^{\frac{8n}{10}}(\frac{(10e) (bN)}{8(r-1)n})^{\frac{8(r-1)n}{10}}(\frac{10eN}{8rn})^{\frac{8n}{10}}e\sqrt{\frac{8n}{10}}(\frac{8n}{10e})^{\frac{8n}{10}}r^{\frac{8n}{10}}e\sqrt{\frac{8(r-1)n}{10}}(\frac{8(r-1)n}{10e})^{\frac{8(r-1)n}{10}}}{(\frac{19}{20}N)^{\frac{8rn}{10}}}
\end{align*}
By some simplifications, we have
\begin{equation}
\label{eq 23}
\frac{|R_2|}{N!}\le (\frac{10e}{8}b^{r-1}(\frac{20}{19})^{r})^{\frac{8n}{10}} e^2\sqrt{\frac{8^2(r-1)n^2}{10^2}}=\exp(-\Theta(n)).
\end{equation}
In the last step, we used that $\frac{10e}{8}b^{r-1}(\frac{20}{19})^{r}$ is a constant strictly smaller than 1. This is true because $b$ is a small constant, say $b\le 0.25$ (see Section~\ref{assumptions}).

To finish the proof, combining Equations~(\ref{eq 21}), (\ref{eq 22}), and (\ref{eq 23}) yields
\[
Pr[\mathcal{E}_3]\ge 1-\exp(\Theta(n))
\]\qed
\end{proof}
\begin{claim}
\label{claim 3}
$Pr[\mathcal{E}_4]\ge 1-\exp(-\Theta(N))$.
\end{claim}
\begin{proof} Define random variable $X^n_t$ to be the number of Believers who become Indifferent in round $t$ given $N_{t-1}=n$. Consider an arbitrary labeling from $1$ to $n$ on the Believers and define Bernoulli random variable $x_i$ for $1\le i\le n$ to be 1 if and only if the $i$-th Believer becomes Indifferent. Recall that each Believer becomes an Indifferent with probability $d$ independently in round $t$, except Seeds who remain Believers. Thus, $Pr[x_i=1]\le d$. Clearly, $X_t^n=\sum_{i=1}^{n}x_i$ and $\mathbb{E}[X_t^n]=(n-s)d\le nd$. Note that $x_i$s are independent; thus, by applying Chernoff bound (see Theorem~\ref{Chernoff}) we get
\begin{align*}
& Pr[X_t^n\ge \frac{5}{40}n]\le Pr[X_t^n\ge (1+\epsilon)dn]\le Pr[X^t_n\ge (1+\epsilon)\mathbb{E}[X_t^n]]\le\\
& \exp(-\frac{\epsilon^2\mathbb{E}[X_t^n]}{3})= \exp(-\Theta(n))
\end{align*}
where we choose the constant $\epsilon>0$ sufficiently small to satisfy $(1+\epsilon)d\le \frac{5}{40}$. This is doable since we assume that $d\le 0.1$, see Section~\ref{assumptions}. Therefore,
\[
Pr[\mathcal{E}_4]= Pr[X_t^n< \frac{5}{40}n]=1-Pr[X_t^n\ge\frac{5}{40}n]=1-\exp(-\Theta(n)).
\]
\qed
\end{proof}

If both events $\mathcal{E}_3$ and $\mathcal{E}_4$ occur, then the number of Believers will increase from $n$ to at least $n+\frac{3}{20}n-\frac{5}{40}n=(1+\frac{1}{40})n$. Since $\mathcal{E}_3$ and $\mathcal{E}_4$ are independent, by applying Claims~\ref{claim 2} and~\ref{claim 3} we get
\begin{equation}
\label{eq 5}
Pr[N_t^n\ge \frac{41}{40}n]\ge Pr[\mathcal{E}_3 \cap \mathcal{E}_4]=Pr[\mathcal{E}_3]\cdot Pr[\mathcal{E}_4]= 1-\exp(-\Theta(n))\ge 1-c_1^{-n}
\end{equation}
for some constant $c_1>1$.

Equation~(\ref{eq 5}) states that in each round the number of Believers increases by a constant fraction with a probability converging to one. Building on this statement and applying the union bound, we show that a.a.s. in logarithmically many rounds the number of Believers overpasses $N^{\frac{2}{3}}$.

Recall that in phase 1 we proved that there is some $t_1\le \log N$ so that $N_{t_1}\ge n_1$ a.a.s. for $n_1=\log_c\log\log N$. 

\begin{claim}
\label{claim 10}
Assume that $(\frac{41}{40})^{t-1}n_1< N^{\frac{2}{3}}$ for some positive integer $t$. Then, given $N_{t_1}\ge n_1$, we have $N_{t_1+t}\ge (\frac{41}{40})^{t}n_1$ with probability at least $1-\sum_{t'=0}^{t-1}c_1^{-(\frac{41}{40})^{t'}n_1}$.
\end{claim}
\begin{proof}
Let us define the events
\begin{align*}
& A_t:=\{N_t\ge (\frac{41}{40})^tn_1\}\quad t\ge 0\\
& A_t':=\{N_t= (\frac{41}{40})^tn_1\}\quad t\ge 0.
\end{align*}
By definition, for $0\le t'\le t$ we have $Pr[A_t|A_{t'}']\le Pr[A_t|A_{t'}]$. We apply this inequality several times later.

The goal is to prove that for $t\ge 1$
\begin{equation*}
Pr[A_t|A_0]\ge 1-\sum_{t'=0}^{t-1}c_1^{-(\frac{41}{40})^{t'}n_1}.
\end{equation*}
We apply proof by induction. As the base case, for $t=1$ we have
\[
Pr[A_1|A_0]\ge Pr[A_1|A_0']\stackrel{\text{Eq.~(\ref{eq 5})}}{\ge} 1-c_1^{-n_1}.
\]
As the induction hypothesis (I.H.), assume that $Pr[A_{t-1}|A_0]\ge 1-\sum_{t'=0}^{t-2}c_1^{-(\frac{41}{40})^{t'}n_1}$. Now, we have
\begin{align*}
& Pr[A_{t}|A_0]\ge Pr[A_{t}|A_0\cap A_{t-1}]\cdot Pr[A_{t-1}|A_0]\ge Pr[A_{t}|A_{t-1}]\cdot Pr[A_{t-1}|A_0]\ge\\ 
& Pr[A_{t}|A_{t-1}']\cdot Pr[A_{t-1}|A_0]\stackrel{\text{Eq.~(\ref{eq 5})}}{\ge} (1-c_1^{-(\frac{41}{40})^{t-1}n_1})\cdot Pr[A_{t-1}|A_0]\stackrel{\text{I.H.}}{\ge}\\
& (1-c_1^{-(\frac{41}{40})^{t-1}n_1}) (1-\sum_{t'=0}^{t-2}c_1^{-(\frac{41}{40})^{t'}n_1})\ge 1-\sum_{t'=0}^{t-1}c_1^{-(\frac{41}{40})^{t'}n_1}.
\end{align*} 
We need the assumption of $(\frac{41}{40})^{t-1}n_1< N^{\frac{2}{3}}$ to apply Equation~(\ref{eq 5}).\qed
\end{proof}

We have that $(\frac{41}{40})^{t}n_1\ge N^{\frac{2}{3}}$ for $t=\log_{\frac{41}{40}}N$ which implies that there exists some $t_2=\mathcal{O}(\log N)$ such that $(\frac{41}{40})^{t_2}n_1\ge N^{\frac{2}{3}}$ and $(\frac{41}{40})^{t_2-1}n_1< N^{\frac{2}{3}}$. Therefore by applying Claim~\ref{claim 10}, given $N_{t_1}\ge n_1$, we have $N_{t_1+t_2}\ge N^{\frac{2}{3}}$ with probability at least
\[
1-\sum_{t'=0}^{t_2-1}c_1^{-(\frac{41}{40})^{t'}n_1}.
\]
To bound this probability, let $t^{\prime\prime}=\log_{\frac{41}{40}}\log_{c_1}\log N$. Then,
\[
\sum_{t'=0}^{t_2-1}c_1^{-(\frac{41}{40})^{t'}n_1}\le \sum_{t'=0}^{t^{\prime\prime}-1}c_1^{-(\frac{41}{40})^{t'}n_1}+\sum_{t'=t^{\prime\prime}}^{t_2-1}c_1^{-(\frac{41}{40})^{t'}n_1}\le t^{\prime\prime}c_1^{-n_1}+t_2c_1^{-(\frac{41}{40})^{t^{\prime\prime}}n_1}.
\]
Since $t_2=\mathcal{O}(\log N)$ and $n_1=\log_c\log\log N$, this is upper-bounded by
\[
\mathcal{O}(\log\log\log N)c_1^{-\log_c\log\log N}+\mathcal{O}(\log N)(\log N)^{-\log_c\log\log N}=o(1)+o(1)=o(1).
\]
Hence, given $N_{t_1}\ge n_1$, we have $N_{t_1+t_2}\ge N^{\frac{2}{3}}$ for some $t_2= \mathcal{O}(\log N)$ with probability $1-o(1)$. 

Overall by phases 1 and 2, starting from $s$ Believers, a.a.s. the process reaches at least $N^{\frac{2}{3}}$ Believers in a logarithmic number of rounds.

\paragraph{Phase 3.} So far we showed that after $t_1+t_2=\mathcal{O}(\log N)$ rounds a.a.s. there are at least $N^{2/3}$ Believers. Now, we prove that from a state with at least $N^{2/3}$ Believers, more than half of the community will become Believers (i.e., the rumor takes over) in logarithmically many rounds. We assume that $n\le N/2$ in this phase.

Let us first lower-bound the value of $\mathbb{E}[N_t^n]$ in this setting. By Equation~(\ref{eq 1}), we have
\[
\mathbb{E}[N_t^n]\ge (1-d)n+((1-b)N-n)\frac{(r-1)!\sum_{j=1}^{r-1}{n\choose j}{N-n-1 \choose r-j-1}}{\prod_{j=1}^{r-1}(N-j)}
\]
To simplify this bound, we compute
\begin{align*}
& \frac{(r-1)!\sum_{j=1}^{r-1}{n\choose j}{N-n-1 \choose r-j-1}}{\prod_{j=1}^{r-1}(N-j)}=\frac{(r-1)!({N-1\choose r-1}-{N-n-1\choose r-1})}{\prod_{j=1}^{r-1}(N-j)}=\\
&\frac{(r-1)!\frac{\prod_{j=1}^{r-1}(N-j)}{(r-1)!}}{\prod_{j=1}^{r-1}(N-j)}-\frac{(r-1)!\frac{\prod_{j=1}^{r-1}(N-n-j)}{(r-1)!}}{\prod_{j=1}^{r-1}(N-j)}=1-\prod_{j=1}^{r-1}\frac{N-n-j}{N-j}\ge\\ 
& 1-\frac{N-n-1}{N-1}\ge 1-\frac{N-n}{N}=\frac{n}{N}.
\end{align*}
Therefore,
\[
\mathbb{E}[N_t^n]\ge (1-d)n+((1-b)N-n)\frac{n}{N}=(2-d-b-\frac{n}{N})n\ge (\frac{3}{2}-d-b)n
\]
where we applied $n\le N/2$ in the last step. Recall that we always assume that $d$ and $b$ are small constants, say $d\le 0.1$ and $b\le 0.25$. (See Section~\ref{assumptions}) Thus, 
\[
\mathbb{E}[N_t^n]\ge (1+\delta) n
\]
for some small constant $\delta>0$. This implies that in expectation the number of Believers is multiplied by $1+\delta$ in each round. Thus, we expect the rumor to take over in a logarithmic number of rounds. In the following, we prove that this expected behavior occurs a.a.s.

By applying the second part of Corollary~\ref{corollary 1} for $\epsilon=\delta/4$ and using $(1-\frac{\delta}{4})(1+\delta)\ge (1+\frac{\delta}{2})$, we have
\begin{align*}
& Pr[N_t^n\le (1+\frac{\delta}{2})n]\le Pr[N_t^n\le (1-\frac{\delta}{4})(1+\delta)n]\le Pr[N_t^n\le (1-\frac{\delta}{4})\mathbb{E}[N_t^n]]\le\\
& \exp(-\Theta(\frac{\mathbb{E}[N_t^n]^2}{N}))\le \exp(-\Theta(\frac{((1+\delta)n)^2}{N}))=\exp(-\Theta(\frac{n^2}{N})).
\end{align*}
Thus, there exists some constant $c_2>1$ such that
\begin{equation}
 \label{eq 25}
Pr[N_t^n\ge (1+\frac{\delta}{2})n]\ge 1-c_2^{-\frac{n^2}{N}}.
\end{equation}
Equation~(\ref{eq 25}) states that given $N_{t-1}=n$, for some $n=\omega(\sqrt{N})$, we have a.a.s. $N_t\ge (1+\frac{\delta}{2})n$. Using this statement and the union bound we prove that a.a.s. after a logarithmic number of round the rumor takes over.
\begin{claim}
\label{claim 11}
Assume that $(1+\frac{\delta}{2})^{t-1}N^{\frac{2}{3}}< \frac{N}{2}$ for some positive integer $t$. Then, given $N_{t_1+t_2}\ge N^{\frac{2}{3}}$, we have $N_{t_1+t_2+t}\ge (1+\frac{\delta}{2})^{t}N^{\frac{2}{3}}$ with probability at least $1-\sum_{t'=0}^{t-1}c_2^{-(1+\frac{\delta}{2})^{2t'}N^{\frac{1}{3}}}$.
\end{claim}
\begin{proof}
The proof is analogous to the proof of Claim~\ref{claim 10} by defining the events $A_t$ and $A_t'$ respectively to be $\{N_t\ge (1+\frac{\delta}{2})^{t}N^{\frac{2}{3}}\}$ and $\{N_t= (1+\frac{\delta}{2})^{t}N^{\frac{2}{3}}\}$. Furthermore, we use Equation~(\ref{eq 25}) instead of Equation~(\ref{eq 5}). Note that we need the condition $(1+\frac{\delta}{2})^{t-1}N^{\frac{2}{3}}< \frac{N}{2}$ to apply Equation~(\ref{eq 25}). \qed 
\end{proof}

Since $(1+\frac{\delta}{2})^tN^{\frac{2}{3}}\ge \frac{N}{2}$ for $t=\log_{1+\frac{\delta}{2}}N$, there exists some $t_3=\mathcal{O}(\log N)$ such that $(1+\frac{\delta}{2})^{t_3}N^{\frac{2}{3}}\ge \frac{N}{2}$ and $(1+\frac{\delta}{2})^{t_3-1}N^{\frac{2}{3}}< \frac{N}{2}$. Thus by applying Claim~\ref{claim 11}, given $N_{t_1+t_2}\ge N^{\frac{2}{3}}$, we have $N_{t_1+t_2+t_3}\ge \frac{N}{2}$ with probability at least
\[
1-\sum_{t'=0}^{t_3-1}c_2^{-(1+\frac{\delta}{2})^{2t'}N^{\frac{1}{3}}}\ge 1-t_3\exp(-\Theta(N^{\frac{1}{3}}))\ge 1-\mathcal{O}(\log N)\exp(-\Theta(N^{\frac{1}{3}}))
\]
which is equal to $1-o(1)$. Thus, given $N_{t_1+t_2}\ge N^{\frac{2}{3}}$ a.a.s $N_{t_1+t_2+t_3}\ge \frac{N}{2}$ for some $t_3=\mathcal{O}(\log N)$.

Overall, combining phases 1, 2, and 3 implies that by starting from $s$ Believers the rumor takes over a.a.s. in $\mathcal{O}(\log N)$ rounds.
\begin{remark}
The error probability in Equation~(\ref{eq 25}) is of the form $\exp(-\Theta(\frac{n^2}{N}))$, which is bounded by a constant for $n=\mathcal{O}(\sqrt{N})$. This should explain why in phases 1 and 2, where $n\le N^{\frac{2}{3}}$, we applied a different approach. However, our argument works for any choice of $N^{\frac{1}{2}+\epsilon'}$ for small $\epsilon'>0$ instead of $N^{\frac{2}{3}}$.
\end{remark}

\paragraph{Tightness.} We claim that the upper bound of $\mathcal{O}(\log N)$ is tight. Note that starting with $s$ Believers, the number of Believers is multiplied by at most an $r$ factor in each round. Thus, it takes the rumor at least $\log_r \frac{N}{2s}=\Omega(\log N)$ to take over.

\subsection{Proof of Part (ii): $\mathbf{m\ge 2}$}
\label{m=2}
In this section, we analyze the $m$-rumor spreading model for $m\ge 2$. We prove that if even initially Believers constitute a constant fraction of the community (i.e., $N_0=\alpha N$ for some small constant $\alpha>0$), the rumor dies out a.a.s. in $\mathcal{O}(\log N)$ rounds. We divide our analysis into two phases. In phase 1, we show that starting from $\alpha N$ Believers, the process reaches a state with less than $F(N,d):=\log_{\frac{1}{d}}\log\log N$ Believers a.a.s. in $\mathcal{O}(\log N)$ rounds. Then in phase 2, we prove that from a state with at most $F(N,d)$ Believers, all agents will become Indifferents, except Seeds, a.a.s. in logarithmically many rounds. We should mention that value of $F(N,d)$ has been chosen to make the calculations straightforward.

\paragraph{Phase 1.} We show that from a state with at most $\alpha N$ Believers, the process reaches at most $F(N,d)$ Believers in logarithmically many rounds a.a.s. In this phase we always assume that $F(N,d)\le n\le \alpha N$. The main idea is to show that in this setting in each round at most $\frac{1}{4}dn$ Indifferents become Believers and at least $\frac{3}{4}dn$ Believers become Indifferents. Thus, the number of Believers decreases by a constant factor in each round. To make this argument more formal, let us define $\mathcal{E}_1$ to be the event that at most $\frac{1}{4}dn$ agents switch from Indifferent to Believer in round $t$ given that $N_{t-1}=n$. Furthermore, let $\mathcal{E}_2$ be the event that the number of Believers who become Indifferents is at least $\frac{3}{4}dn$ in round $t$ given $N_{t-1}=n$. 

We first bound the probabilities $Pr[\mathcal{E}_1]$ and $Pr[\mathcal{E}_2]$ respectively in Claim~\ref{claim 4} and Claim~\ref{claim 5}. Building on them, we provide a lower bound on the probability $Pr[N_t^n\ge (1-\frac{d}{2})n]$.
\begin{claim}
\label{claim 4}
$Pr[\mathcal{E}_1]\ge 1-\exp(-\Theta(n))$.
\end{claim} 
\begin{proof}
We claim that
\begin{align*}
Pr[\overline{\mathcal{E}_1}]\le \frac{{\frac{N}{r}\choose \frac{dn}{4r}}\ {n\choose \frac{dn}{2r}}\ \frac{(\frac{dn}{2r})!}{2^{\frac{dn}{4r}}}\ (r(r-1))^{\frac{dn}{4r}}\ (N-\frac{dn}{2r})!}{N!}
\end{align*}
where $\overline{\mathcal{E}_1}$ is the complement of event $\mathcal{E}_1$. The numerator is an upper bound on the number of permutations which result in the generation of at least $\frac{1}{4}dn$ new Believers (Indifferents who become Believers) out of all possible $N!$ permutations. Note that in such permutations, there must exist $\frac{dn}{4r}$ rooms such that each of them contains at least two Believers, otherwise less than $\frac{1}{4}dn$ Indifferents become Believers for $m\ge 2$. There are ${\frac{N}{r}\choose \frac{dn}{4r}}$ possibilities to select $\frac{dn}{4r}$ rooms and ${n\choose \frac{dn}{2r}}$ ways to choose $\frac{dn}{2r}$ Believers. We claim that the number of possibilities to assign to each of these rooms exactly two of the Believers is equal to $(\frac{dn}{2r})!/2^{\frac{dn}{4r}}$. This is true because the number of choices to pair $\frac{dn}{4r}$ agents is $\frac{(\frac{dn}{2r})!}{(\frac{dn}{4r})!\ 2^{\frac{dn}{4r}}}$ (see e.g.~\cite{andrews1998theory}) and there are $(\frac{dn}{4r})!$ ways of distributing these pairs into the rooms; thus, the number of all possibilities is equal to
\[
\frac{(\frac{dn}{2r})!}{(\frac{dn}{4r})!\ 2^{\frac{dn}{4r}}}(\frac{dn}{4r})!=\frac{(\frac{dn}{2r})!}{2^{\frac{dn}{4r}}}.
\]
Finally, there are $r(r-1)$ possible placements for two Believers in a room of size $r$ and there are $(N-\frac{dn}{2r})!$ possible placements for the remaining agents. In this way, we clearly count all our desired permutations at least once (we are actually over-counting, but this is fine since we need an upper bound).

By ${n \choose k}\le (\frac{ne}{k})^k$, Stirling's approximation~\cite{marsaglia1990new}, we get
\begin{align*}
& Pr[\overline{\mathcal{E}_1}]\le \frac{(\frac{4eN}{dn})^{\frac{dn}{4r}}\ (\frac{2er}{d})^{\frac{dn}{2r}}\ (\frac{dn}{2r})^{\frac{dn}{2r}}\ r^{\frac{dn}{2r}}}{2^{\frac{dn}{4r}}\ \prod_{j=0}^{\frac{dn}{2r}-1}(N-j)}\le \frac{(\frac{2eN}{dn})^{\frac{dn}{4r}}e^{\frac{dn}{2r}}n^{\frac{dn}{2r}}r^{\frac{dn}{2r}}}{(\frac{9}{10}N)^{\frac{dn}{2r}}}=\\
& (\frac{2e^3r^210^2}{9^2d})^{\frac{dn}{4r}}(\frac{n}{N})^{\frac{dn}{4r}}\le (\frac{2e^3r^210^2}{9^2d})^{\frac{dn}{4r}}\alpha^{\frac{dn}{4r}}=\exp(-\Theta(n)).
\end{align*}
where we used $n/N\le \alpha$. Furthermore, we applied twice that $\alpha$ is a sufficiently small constant. Firstly, we have $(N-j)\ge \frac{9}{10}N$ for $0\leq j\leq \frac{dn}{2r}-1$ by utilizing $\frac{dn}{2r}-1\le \frac{1}{10}N$, which is true for $\alpha\le \frac{2r}{10d}=\frac{r}{5d}$. Secondly, we applied $\alpha<\frac{9^2d}{2e^3r^210^2}$ in the last step. Thus, we need constant $\alpha$ to be smaller than $\min (\frac{r}{5d},\frac{9^2d}{2e^3r^210^2})$, which holds since we assume $\alpha$ to be a sufficiently small constant. \qed
\end{proof}

\begin{claim}
\label{claim 5}
$Pr[\mathcal{E}_2]\ge 1-\exp(-\Theta(n))$.
\end{claim} 
\begin{proof}
Define random variable $X^n_t$ to be the number of Believers who switch into Indifferents in round $t$ given $N_{t-1}=n$. Recall that each Believer becomes Indifferent with probability $d$ independently, except Seeds. Consider an arbitrary labeling from $1$ to $n$ on the Believers and define Bernoulli random variable $x_i$ for $1\le i\le n$ to be 1 if and only if the $i$-th Believer becomes Indifferent. We know that $Pr[x_i=1]$ is equal to $d$ if the $i$-the agent is Other and 0 if it is a Seed. Clearly, $X_t^n=\sum_{i=1}^{n}x_i$, which implies that $\mathbb{E}[X_t^n]=(n-s)d$. Thus, $\frac{9}{10}dn\le \mathbb{E}[X^t_n]$ by using $sd\le \frac{1}{10}n$, which is true since $s$ is a small constant and $n\ge F(N,d)= \log_{\frac{1}{d}}\log\log N$. Notice that $x_i$s are independent; thus, by applying Chernoff bound (see Theorem~\ref{Chernoff}) we have
\begin{align*}
& Pr[\mathcal{E}_2]=Pr[\frac{3}{4}dn\le X_t^n]=Pr[(1-\frac{1}{6})\frac{9}{10}dn\le X_t^n]\ge Pr[(1-\frac{1}{6})\mathbb{E}[X_t^n]\le X_t^n]=\\
& 1-Pr[X_t^n< (1-\frac{1}{6})\mathbb{E}[X_t^n]]\ge 1-\exp(-\frac{(\frac{1}{6})^2\mathbb{E}[X_t^n]}{2})\ge 1-\exp(-\frac{(\frac{1}{6})^2\frac{9}{10}dn}{2})
\end{align*}
which is equal to $1-\exp(-\Theta(n))$. \qed
\end{proof}

If both events $\mathcal{E}_1$ and $\mathcal{E}_2$ occur, at most $\frac{1}{4}dn$ Indifferents become Believers and at least $\frac{3}{4}dn$ Believers become Indifferents, which implies that the number of Believers decreases from $n$ to $(1-\frac{d}{2})n$. Since, events $\mathcal{E}_1$ and $\mathcal{E}_2$ are independent, applying Claims~\ref{claim 4} and~\ref{claim 5} yields
\begin{align*}
& Pr[N_t^n\le (1-\frac{d}{2})n]\ge Pr[\mathcal{E}_1\cap \mathcal{E}_2]=Pr[\mathcal{E}_1]\cdot Pr[\mathcal{E}_2]\ge\\ &(1-\exp(-\Theta(n)))\ (1-\exp(-\Theta(n)))= 1-\exp(-\Theta(n)).
\end{align*}
Therefore, there exists some constant $c>1$ such that for $F(N,d)\le n\le \alpha n$, we have
\begin{equation}
\label{eq 13}
Pr[N_t^n\le (1-\frac{d}{2})n]\ge 1-c^{-n}.
\end{equation}
This implies that if in round $t-1$ there are $n$ Believers a.a.s. in the $t$-th round there are at most $(1-\frac{d}{2})n$ Believers. Building on this statement and applying the union bound, we prove that from $N_0\le \alpha N$, the number of Believers decreases to $F(N,d)$ a.a.s. in $\mathcal{O}(\log N)$ rounds. To make this argument more formal let us provide Claim~\ref{claim 6}.

\begin{claim}
\label{claim 6}
Assume that $(1-\frac{d}{2})^{t-1}\alpha N\ge F(N,d)$ for some positive integer $t$. Then,
\[
Pr[N_t\le (1-\frac{d}{2})^t\alpha N|N_0\le\alpha N]\ge 1-\sum_{t'=0}^{t-1}c^{-(1-\frac{d}{2})^{t'}\alpha N}.
\]
\end{claim}
\begin{proof}
We apply an inductive argument similar to the one in the proof of Claim~\ref{claim 10}. Let us define the events
\begin{align*}
& A_t:=\{N_t\le (1-\frac{d}{2})^t\alpha N\}\quad t\ge 0\\
& A_t':=\{N_t= (1-\frac{d}{2})^t\alpha N\}\quad t\ge 0.
\end{align*}
Note that by definition, for $0\le t'\le t$ we have $Pr[A_t|A_{t'}'] \le Pr[A_t|A_{t'}]$, which we are going to utilize several times later.

Our goal is to prove that for $t\ge 1$
\begin{equation*}
\label{eq 14}
Pr[A_t|A_0]\ge 1-\sum_{t'=0}^{t-1}c^{-(1-\frac{d}{2})^{t'}\alpha N}.
\end{equation*}
Now, we do induction on $t$. As the base case, for $t=1$
\[
Pr[A_1|A_0]\ge Pr[A_1|A_0']\stackrel{\text{Eq.~(\ref{eq 13})}}{\ge} 1-c^{-\alpha N}.
\]
As the induction hypothesis (I.H.), assume that $Pr[A_{t-1}|A_0]\ge 1-\sum_{t'=0}^{t-2}c^{-(1-\frac{d}{2})^{t'}\alpha N}$. Now, we have
\begin{align*}
& Pr[A_{t}|A_0]\ge Pr[A_{t}|A_0\cap A_{t-1}]\cdot Pr[A_{t-1}|A_0]\ge Pr[A_{t}|A_{t-1}]\cdot Pr[A_{t-1}|A_0]\ge\\ 
& Pr[A_{t}|A_{t-1}']\cdot Pr[A_{t-1}|A_0]\stackrel{\text{Eq.~(\ref{eq 13})}}{\ge} (1-c^{-(1-\frac{d}{2})^{t-1}\alpha N})\cdot Pr[A_{t-1}|A_0]\stackrel{\text{I.H.}}{\ge}\\
& (1-c^{-(1-\frac{d}{2})^{t-1}\alpha N}) (1-\sum_{t'=0}^{t-2}c^{-(1-\frac{d}{2})^{t'}\alpha N})\ge 1-\sum_{t'=0}^{t-1}c^{-(1-\frac{d}{2})^{t'}\alpha N}.
\end{align*} 
Notice that we need the condition $(1-\frac{d}{2})^{t-1}\alpha N\ge F(N,d)$ to apply Equation~(\ref{eq 13}). \qed
\end{proof}

Now, by applying Claim~\ref{claim 6} for a suitable choice of $t$, we finish the proof of phase 1. Let us define $$T:=\lceil \log_{(1-\frac{d}{2})}\frac{F(N,d)}{\alpha N}\rceil.$$ We have that $(1-\frac{d}{2})^{T-1}\alpha N\ge F(N,d)$ because
\begin{equation}
\label{eq 15}
(1-\frac{d}{2})^{T-1}\ge (1-\frac{d}{2})^{\log_{(1-\frac{d}{2})}\frac{F(N,d)}{\alpha N}}=\frac{F(N,d)}{\alpha N}.
\end{equation} 
Thus by Claim~\ref{claim 6}, given $N_0\le \alpha N$, we have
\[
N_T\le (1-\frac{d}{2})^T\alpha N\le F(N,d)
\]
with probability at least
\[
1-\sum_{t'=0}^{T-1}c^{-(1-\frac{d}{2})^{t'}\alpha N}.
\]
We show that this probability is equal to $1-o(1)$ and $T=\mathcal{O}(\log N)$. Therefore, there will exist at most $F(N,d)=\log_{\frac{1}{d}}\log\log N$ Believers a.a.s. after $\mathcal{O}(\log N)$ rounds.

Since $(1-\frac{d}{2})=(\frac{2}{2-d})^{-1}$,
\[
T=\lceil \log_{(1-\frac{d}{2})}\frac{F(N,d)}{\alpha N}\rceil= \lceil\log_{\frac{2}{2-d}}\frac{\alpha N}{F(N,d)}\rceil\le\lceil \log_{\frac{2}{2-d}}N\rceil=\mathcal{O}(\log N).
\]
It remains to show that $\sum_{t'=0}^{T-1}c^{-(1-\frac{d}{2})^{t'}\alpha N}=o(1)$. We have
\begin{align*}
& \sum_{t'=0}^{T-1}c^{-(1-\frac{d}{2})^{t'}\alpha N}\le \sum_{t'=0}^{T-1}\frac{1}{(1-\frac{d}{2})^{t'}\alpha N}=\frac{1}{(1-\frac{d}{2})^{T-1}\alpha N}\sum_{t'=0}^{T-1}\frac{1}{(1-\frac{d}{2})^{t'-(T-1)}}\stackrel{\text{Eq.~(\ref{eq 15})}}{\le}\\
& \frac{1}{F(N,d)}\sum_{t'=0}^{T-1}(1-\frac{d}{2})^{t'}\le \frac{1}{F(N,d)}\sum_{t'=0}^{\infty}(1-\frac{d}{2})^{t'}=\frac{2}{F(N,d)d}=o(1).
\end{align*}
In the second-to-last step we used that the geometric series $\sum_{t'=0}^{\infty}(1-\frac{d}{2})^{t'}$ is equal to $\frac{1}{1-(1-d/2)}=\frac{2}{d}$.

\paragraph{Phase 2.} So far we proved that if $N_0\le \alpha N$, then $N_{t_1}\le F(N,d)=\log_{\frac{1}{d}}\log\log N$ a.a.s. for some $t_1=\mathcal{O}(\log N)$. In this phase, we show that from a state with at most $F(N,d)$ Believers, there is no Believer, except Seeds, in a logarithmic number of rounds a.a.s.

Let us first prove Claim~\ref{claim 7}, which asserts if in some state the number of Believers is at most $F(N,d)$, their number does not overpass $F(N,d)$ a.a.s. during the next $\log N$ rounds.
\begin{claim}
\label{claim 7}
Given $N_{t_1}\le F(N,d)$, we have $N_{t_1+t}\le F(N,d)$ for $1\leq t\le \log N$ a.a.s.
\end{claim}
\begin{proof}
Let us define for $0\le t\le \log N$ the event $B_t:=\{N_{t_1+t}\le F(N,d)\}$. We want to prove that $Pr[\bigwedge_{t=1}^{\log N}B_t|B_0]=1-o(1)$. We have
\[
Pr[\bigwedge_{t=1}^{\log N}B_t|B_0]=\prod_{t=1}^{\log N}Pr[B_t|\bigwedge_{t'=0}^{t-1}B_{t'}]
\] 
where the right-hand side is a telescoping product. We prove that
\begin{equation}
\label{eq 2}
Pr[B_t|\bigwedge_{t'=0}^{t-1}B_{t'}]\ge 1-\frac{1}{\sqrt{N}}.
\end{equation}
Therefore,
\[
Pr[\bigwedge_{t=1}^{\log N}B_t|B_0]\ge (1-\frac{1}{\sqrt{N}})^{\log N}\ge 1-\frac{\log N}{\sqrt{N}}=1-o(1).
\]
It remains to prove that Equation~(\ref{eq 2}) holds. Let $\overline{B}_t$ be the complement of $B_t$. We claim that
\[
Pr[\overline{B}_t|N_{t_1+t-1}\le F(N,d)]\le \frac{\frac{N}{r}{F(N,d)\choose 2}{r\choose 2}(N-2)!}{N!}.
\]
Given $N_{t_1+t-1}\le F(N,d)$, for event $\overline{B}_t$ to happen (i.e., $N_{t_1+t}>F(N,d)$), at least two Believers must share a room in round $t_1+t$, to generate a new Believer. The numerator is an upper bound on the number of permutations in which at least two Believers share a room, from all $N!$ possible permutations. There are $N/r$ possibilities to select a room and at most ${F(N,d)\choose 2}$ possibilities to choose two Believers. There are ${r\choose 2}$ ways to place the two Believers in the room and $(N-2)!$ possibilities to place the remaining $N-2$ agents. We count each permutation with at least two Believers in a room at least once in this way. We might actually count a permutation several times, but it does not matter since we are interested in an upper bound.

Simplifying the right-hand side gives us
\[
Pr[\overline{B}_t|N_{t+t_1-1}\le F(N,d)]\le \frac{F(N,d)^2r}{N-1}
\]
which is smaller than $\frac{1}{\sqrt{N}}$ since $F(N,d)=\log_{\frac{1}{d}}\log\log N$. Thus, we get $Pr[B_t|N_{t+t_1-1}\le F(N,d)]\ge 1-\frac{1}{\sqrt{N}}$. By applying the definition of $B_t$, we have
\[
Pr[B_t|\bigwedge_{t'=0}^{t-1}B_{t'}]\ge Pr[B_t|N_{t+t_1-1}\le F(N,d)]\ge 1-\frac{1}{\sqrt{N}}
\]
which finishes the proof. \qed
\end{proof}

By Claim~\ref{claim 7}, given $N_{t_1}\le F(N,d)$, we have a.a.s. $N_{t_1+t}\le F(N,d)$ for $1\le t\le \log N$, that is, for $\log N$ rounds the number of Believers will not exceed $F(N,d)$. Furthermore, in each round a Believer becomes Indifferent independently with probability $d$, except Seeds. Thus, in each of these $\log N$ rounds the probability that all Believers in that round (whose number is at most $F(N,d)$) become Indifferents is at least
\[
d^{\log_{\frac{1}{d}}\log\log N}=\frac{1}{\log\log N}.
\]
The probability that this does not occur in any of these $\log N$ rounds is upper-bounded by
\[
(1-\frac{1}{\log\log N})^{\log N}\le \exp(-\frac{\log N}{\log\log N})=o(1)
\]
where we used $1-x\le \exp(-x)$. Thus, a.a.s. after at most $t_1+\log N=\mathcal{O}(\log N)$ rounds, the rumor dies out.

\paragraph{Tightness.} We claim that the upper bound of $\mathcal{O}(\log N)$ is tight. That is, if the process starts with $\alpha N$ Believers, it takes the rumor at least a logarithmic number of rounds to die out. Assume that $N_0=\alpha N$ and no new Believer is generated during the process, i.e., no Indifferent becomes a Believer. Even under this assumption, in each round in expectation a $d$ fraction of the Believers will become Indifferents. Thus in expectation, it takes the process logarithmically many rounds to reach the state where all agents are Indifferents, except Seeds. To turn this argument into an a.a.s. statement, we can apply Chernoff bound (see Theorem~\ref{Chernoff}). Assume $X_t^n$ denote the number of Believers who become Indifferent in the $t$-th round, given $N_{t-1}=n$. Since each Believer becomes Indifferent independently with probability $d$, except $s$ Seeds, then $\mathbb{E}[X_t^n]=(n-s)d\le nd$. Thus, for a small constant $\epsilon>0$ we have
\begin{align*}
Pr[(1+\epsilon)nd\le X_t^n]\le Pr[(1+\epsilon)\mathbb{E}[X_t^n]\le X_t^n]\le \exp(-\Theta(n)).
\end{align*}
Thus, if there are $n=\omega(1)$ Believers in round $t-1$, a.a.s. in round $t$ at most $(1+\epsilon)d$ fraction of them will become Indifferents. Now, applying an inductive argument similar to the one from Claim~\ref{claim 6}, it is easy to show that starting from $N_0=\alpha N$, a.a.s. it takes the rumor $\Omega(\log N)$ rounds to die out. 

\section{Illustrations of the model} 
\label{illustration}
As we discussed, for the purpose of the present paper it is realistic to assume that parameters $b$ and $d$ are relatively small constants, say $b\le 0.25$ and $d\le 0.1$. Furthermore, since the behavior of the process is a function of several parameters, it is sensible to fix some parameters to study the impact of the others in a more transparent set up. The range that the value of a fixed parameter is chosen from can be set according to the potential application of the model.

Nevertheless, it would be interesting to study the $m$-rumor spreading model for larger values of $d$ and $b$ in future work. In this section, we try to illustrate the behavior of the process in this setting building on some approximate arguments. However, the rigorous analysis is left for future work.

We first introduce some notations. Let the \emph{density} of Believers be the ratio of the number of Believers to $N$. Furthermore, let $f(P)$ denote the expected number of Believers in the next round by assuming that the density of Believers in the current round is equal to $P$, i.e., there are $PN$ Believers. 

Let us first consider the case of $m=1$ and $r=2$. Note that each Believer becomes Indifferent independently with probability $d$. Thus in expectation a $d$ fraction of Believers will become Indifferents. (Actually, there are $s$ Seeds who remain Believers regardless of the choice of $d$. However, since $s$ is a small constant, this can be ignored for our rough argument here.) Furthermore, Others who are Indifferents constitute a $(1-b-P)$ fraction of the community. Each of them shares its room with a Believer, and consequently will become a Believer, with probability $P$. Therefore, we can ``approximate'' the value of $f(P)$ to be
\[
f(P)=(1-d)P+(1-b-P)P.
\]
Function $f(P)$ has a unique fixed point $1-b-d$ for $P\in(0,1-b)$. More accurately, $f(P)<P$ for $P\in (0,1-b-d)$ and $P<f(P)$ for $P\in(1-b-d,1-b)$. See Figure~\ref{fig4} (left). Therefore, we expect the process to converge to a state with Believer density almost $1-b-d$. It is easy to check that this is consistent with our results in Theorem~\ref{theorem 1} part (i), which holds for $b\le 0.25$ and $d\le 0.1$.
\begin{figure}[!ht]
\centering
\includegraphics[scale=0.5]{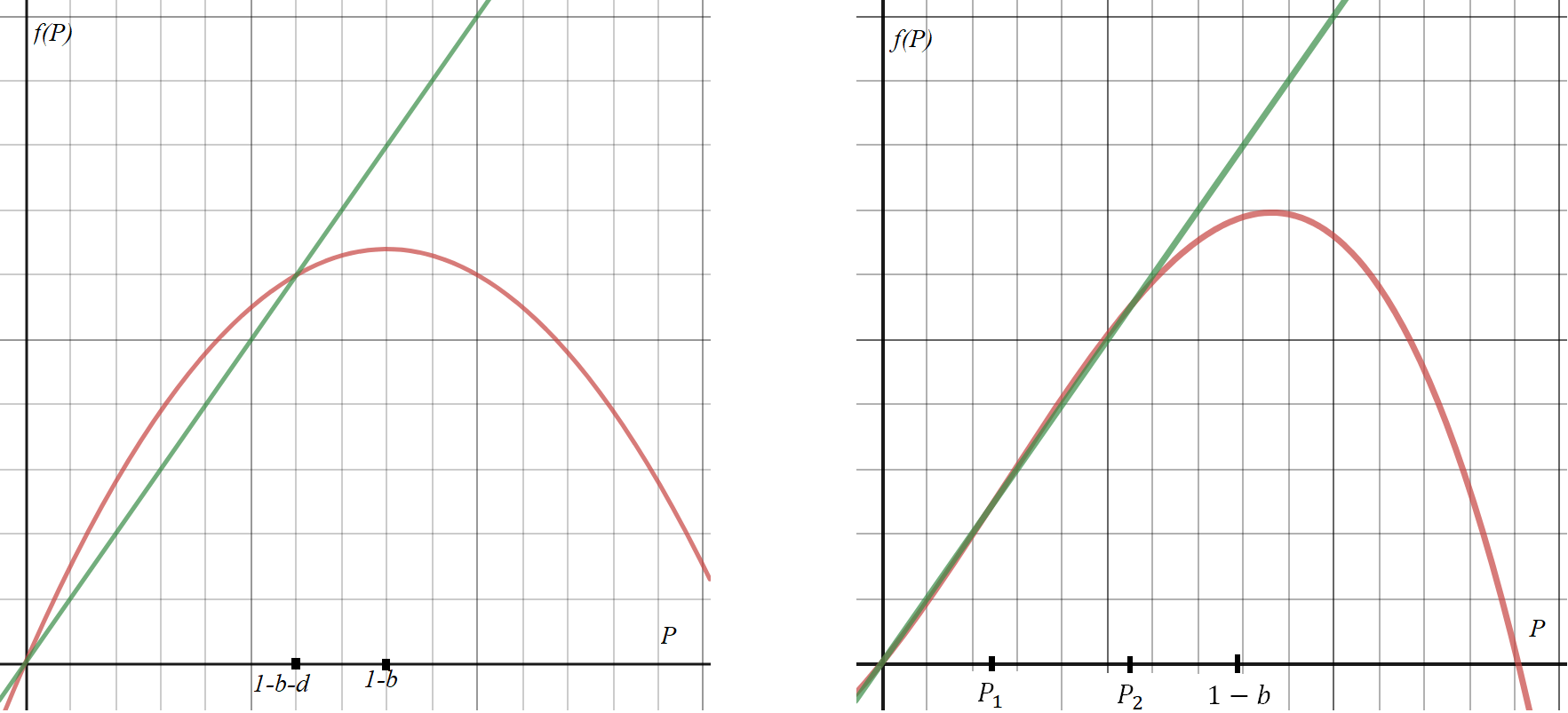}
\caption{(left) function $f(P)$ has a unique fixed point for $P\in(0,1-b)$ (right) function $f(P)$ has two fixed points $P_1$ and $P_2$ for $P\in(0,1-b)$.}
\label{fig4}
\end{figure}

Now, we consider the case of $m=2$ and $r=3$. Note that each Believer becomes Indifferent with probability $d$. (We are again skipping Seeds) Furthermore, an Other who is Indifferent will become Believer if it shares its room with two Believers. Thus, by an argument similar to above, we can ``approximate'' $f(P)$ to be
\[
f(P)=(1-d)P+(1-b-P)P^2.
\]
(This is just an approximation since we assume that each seat is occupied by a Believer with probability $P$ \emph{independently}, which is not true.)

Function $f(P)$ in this setting has two fixed points for $P\in (0,1-b)$. More precisely, $f(P)<P$ for $P\in (0,P_1)$, $f(P)>P$ for $P\in (P_1,P_2)$ and $f(P)<P$ for $P\in(P_2,1-b)$, where
\[
P_1=\frac{(1-b)+\sqrt{(1-b)^2-4d}}{2},\quad
P_2=\frac{(1-b)+\sqrt{(1-b)^2-4d}}{2}.
\]
Therefore, if the initial density of Believers is less than $P_1$, we expect the rumor to die out. However, if the initial density is larger than $P_1$, we expect the process to converge to a state with Believer density of almost $P_2$. See Figure~\ref{fig4} (right).

\section{Conclusion}
\label{conclusion}
We have shown that within the $m$-rumor spreading model, switching from $m=1$ to $m\ge 2$ triggers a drastic qualitative change in the spreading process. More precisely, when $m=1$ even a small group of Believers manage to convince a large part of the community very quickly, but for $m\ge 2$ even a substantial fraction of Believers may not prevent the rumor from dying out after a few update rounds. Our results shed a light on why a given rumor spreads within a social group and not in another as noted in Section~\ref{real cases} regarding some recent real cases of rumor spreading.

At this stage it is worth stressing that our model like other models, is not the reality. However, the common aim is to grasp some essential traits of the reality despite making crude approximations. We hope our findings shed some additional light on the understanding of rumor spreading phenomena.

\paragraph{Acknowledgment.} The first author is thankful to Bernd G\"artner for several stimulating conversations on a primary variant of the model.
\bibliographystyle{acm} %
\bibliography{references}
\end{document}